\RequirePackage{fix-cm}
\documentclass{amsart}       

\usepackage{amsaddr}
\usepackage{tikz}
\usepackage{amsmath}
\usepackage{amsfonts}
\usepackage{amssymb}
\usepackage{amsthm}
\usepackage{epstopdf}
\usepackage{graphicx}
\usepackage{enumerate}
\usepackage{xfrac}
\usepackage{amstext}    
\usepackage{array} 
\usepackage[margin=4cm]{geometry}
\usepackage{hyperref}

\usepackage{bm}

\newcommand{\ab}{\mathbf{a}}
\newcommand{\pb}{\mathbf{p}}
\newcommand{\kb}{\mathbf{k}}
\newcommand{\lb}{\mathbf{l}}
\newcommand{\jb}{\mathbf{j}}
\newcommand{\ib}{\mathbf{i}}
\newcommand{\nb}{\mathbf{n}}

\newcommand{\vb}{\mathbf{v}}
\newcommand{\bb}{\mathbf{b}}

\newcommand{\Gb}{\mathbf{\Gamma}}
\newcommand{\Rb}{\mathbb{R}}
\newcommand{\Cb}{\mathbb{C}}
\newcommand{\Zb}{\mathbb{Z}}
\newcommand{\Mcal}{\mathcal{M}}
\newcommand{\aspect}{\kappa}
\newcommand{\xb}{\mathbf{x}}

\newcommand{\omb}{\boldsymbol{\omega}}
\newcommand{\hel}{h}

\newcommand{\zeroS}{\{ \mathbf{0} \}}

\newcommand{\aspm}{{ K} }
\newcommand{\transp}{\mathsf{T}}
\newcommand\reallywidehat[1]{%
\savestack{\tmpbox}{\stretchto{%
  \scaleto{%
    \scalerel*[\widthof{\ensuremath{#1}}]{\kern-.6pt\bigwedge\kern-.6pt}%
    {\rule[-\textheight/2]{1ex}{\textheight}}
  }{\textheight}%
}{0.5ex}}%
\stackon[1pt]{#1}{\tmpbox}%
}

\newcommand{\crossmat}[1]
{\widehat{#1}}

\numberwithin{equation}{section}
\numberwithin{figure}{section}

\usetikzlibrary{patterns}

\theoremstyle{plain}
          \newtheorem{theorem}{Theorem}[section]
          \newtheorem{proposition}{Proposition}[section]
          \newtheorem{lemma}[theorem]{Lemma}

          \theoremstyle{definition}

\begin{document}
\title[Poisson Structure of 3D Euler]
		{Poisson Structure of the Three-Dimensional Euler Equations in Fourier Space}

\author{H.R. Dullin$^{1}$, J.D. Meiss$^{2}$, and  J. Worthington$^{1,3}$}
\address[foot]{ (1) School of Mathematics and Statistics, The University of Sydney, Sydney, NSW 2006, Australia\\ 
                (2) Department of Applied Mathematics, University of Colorado, Boulder, CO 80309-0526, USA \\
			(3) Cancer Research Division, Cancer Council NSW, NSW, Australia}

\date{\today}

\begin{abstract}

We derive a simple Poisson structure in the space of Fourier modes for the vorticity formulation of the Euler equations on a 
three-dimensional periodic domain. This allows us to analyse the structure of the Euler equations using a Hamiltonian framework. 
The Poisson structure is valid on the divergence free subspace only, and we show that using a projection operator
it can be extended to be valid in the full space. We then restrict the simple Poisson structure to the
divergence-free subspace on which the dynamics of the Euler equations take place, 
reducing the size of the system of ODEs by a third. 
The projected and the restricted Poisson structures are shown to have the helicity as a Casimir invariant. 
We conclude by showing that periodic shear flows in three dimensions are equilibria that correspond to singular 
points of the projected Poisson structure, and hence that the usual approach to study their nonlinear stability 
through the Energy-Casimir method fails.

\end{abstract}

\maketitle

\section{Introduction}

The dynamics of the Euler fluid equations has been a rich area of study, a recent review of which is given in \cite{Constantin07}. 
The Hamiltonian nature of hydrodynamical flows and the Poisson structure of the Euler equations has been described 
by \cite{Arnold78, Kuznetsov80, Zakharov97}, and studied extensively in, for example, \cite{Salmon88, Arnold98,Morrison98}.
For more background on fluid brackets and in particular on the question on how to project onto the divergence-free 
subspace see \cite{Chandre13}.

For the case of a two-dimensional, periodic domain, Arnold described the dynamics of the Euler equations as geodesics on the 
manifold of volume-preserving diffeomorphisms \cite{Arnold66,Arnold98}. This was used to obtain, using the energy-Casimir 
method, proofs of stability for examples of shear flows \cite{Arnold69,Arnold98}. 
This structure also led to the truncated Poisson structure approximating the Euler equations that was developed by 
Zeitlin \cite{Zeitlin90,Zeitlin05}, which in turn led to the development of a Poisson integrator \cite{Mclachlan93} and a number of 
stability and instability results for shear flows \cite{DMW16,Dullin2016}. Geometric integration of incompressible fluid 
equations can also be done in physical space, see, e.g. \cite{Salmon04,Pavlov11,Morrison17}. 

We present here three Poisson structures in Fourier space for the Euler equations on a three-dimensional, periodic domain. 
This first is only valid on the divergence free subspace on which the dynamics of the 
Euler equations takes place. Following ideas of \cite{Chandre13} we project this structure to make it valid in the full space
of Fourier coefficients, not just on the divergence free subspace. We then obtain a third version of the Poisson structure by restricting 
it to the divergence free subspace. Both the projected and restricted Poisson structures are shown to have the helicity as a Casimir.
These Poisson structures encapsulate the geometry underlying the dynamics of the Fourier coefficients of the vorticity. 
It is hoped that this will lay the foundations for geometric integration of the three-dimensional incompressible fluid  equations. 
As an application we study periodic shear flows in three dimensions, 
and show that for these equilibria  the  Poisson structure is singular.
This shows that the Energy-Casimir method cannot be applied for these equilibria.

\section{Vorticity Formulation for the Euler Equations}

The vorticity form of the Euler equations for an incompressible, inviscid fluid on some three-dimensional domain 
$\mathcal{D}\subset\Rb^3$ is
\begin{equation}
\frac{\partial \Omega}{\partial t} =  \nabla \times (V \times \Omega) .
\label{eq:Lie}
\end{equation}
Here  $V:\mathcal{D}\times \Rb \to\Rb^3$ describes the instantaneous velocity of a fluid element at position $\xb =(x,y,z)\in\mathcal{D}$ and time $t \in \Rb$, and the vorticity $\Omega(\xb;t)=\nabla\times V$
measures the rotational motion of a fluid element. One can think of each component of $\Omega$ representing the rotation of the fluid around the axis parallel to the corresponding coordinate.
The mathematical structure of these equations is best understood when 
$\Omega$ is thought of as an exact 2-form whose time derivative is given by the Lie derivative of $\Omega$ in the direction of $V$.

We assume that the fluid density $\rho$ is constant --- without loss of generality we can set $\rho = 1$ --- and that
the two vector fields satisfy the \emph{divergence-free conditions} 
\begin{equation}\label{eq:divergence}
	\nabla \cdot V=0,\;\; \nabla \cdot \Omega=0.
\end{equation}
For the velocity field, this is a consequence of the assumed incompressibility; for the vorticity, this follows from the fact that the divergence of the curl of any vector must be zero.
When both fields are divergence free, \eqref{eq:Lie} reduces to
\begin{equation}
\frac{\partial \Omega}{\partial t}=(\Omega \cdot \nabla)V-(V \cdot \nabla ) \Omega.
\label{eq:omvpde}
\end{equation}

We will consider these equations on the periodic domain 
\[
	\mathcal{D}=\left [0,\frac{2\pi}{\aspect_x}\right )\times\left [0,\frac{2\pi}{\aspect_y}\right )\times\left [0,\frac{2\pi}{\aspect_z}\right )
\]
for the wavenumber basis $\aspect_x,\aspect_y,\aspect_z\in\Rb^+$. 
The \emph{anisotropy matrix}
\[
	\aspm:=\begin{pmatrix}
		\aspect_x & 0 & 0 \\
		0 & \aspect_y &  0 \\
		0 & 0 & \aspect_z 
		\end{pmatrix}
\]
represents the unit wavenumbers in $\mathcal{D}$.
For an isotropic domain we can scale so that $\aspm=\mathbb{I}_3$; in general we allow for anisotropy by letting the wavenumbers be arbitrary.
The allowed wavenumbers on $\mathcal{D}$ are elements of the scaled integer lattice
\begin{equation}\label{eq:lattice}
	D=\{ \aspm\ab \; | \; \ab\in \Zb^3\}.
\end{equation}
Compare this to the approach in \cite{Dullin2016}, where the lattice $\Zb^3$ is used for the mode numbers instead of 
$D$. By scaling the lattice instead, the influence of the domain size is mostly hidden from view. This should not be misinterpreted as meaning it is unimportant; for instance,  the stability results of \cite{Dullin2016} depend on the domain size. However, ``hiding'' the aspect ratios in this way makes the exposition  less cumbersome.

The velocity and vorticity Fourier coefficients for a wavenumber $\jb$ in the lattice \eqref{eq:lattice} are given by  integrals over the domain $\mathcal{D}$:
\[
	\begin{split}
	\vb_\jb  &= \frac{\aspect_x\aspect_y\aspect_z}{(2\pi)^3} \int_\mathcal{D} V(\xb,t)e^{-i\jb \cdot \xb }\mathrm{d}\xb, \\
	\omb_\jb &=\frac{\aspect_x\aspect_y\aspect_z}{(2\pi)^3} \int_\mathcal{D} \Omega(\xb,t)e^{-i\jb \cdot \xb}\mathrm{d}\xb ,
	\end{split}
\]
where $\vb_\jb(t),\omb_\jb(t)\in\mathbb{C}^3$. The inverse transform is 
\begin{equation}\label{eq:3Dfourier}
	\begin{split}
	V      &=\sum_{\jb\in D} \vb_\jb(t)e^{i\jb \cdot\xb},\\
	\Omega &=\sum_{\jb\in D} \omb_\jb(t)e^{i\jb\cdot\xb}.
	\end{split}
\end{equation}
Since $V$ and $\Omega$ are real, $\vb_{-\jb}=\overline{\vb}_\jb$ and $\omb_{-\jb}=\overline{\omb}_\jb$. 

The divergence-free conditions \eqref{eq:divergence} imply that for all $j \in D$,
\begin{equation}
	\label{eq:divfree}
	\jb\cdot \vb_\jb =0, \;\;\jb\cdot \omb_\jb =0.
\end{equation}
Thus, all the dynamics of the Euler equations occur on the \emph{divergence-free subspace}
\begin{equation}
\label{eq:divfreess}
	\mathcal{M}=\{ \omb_\jb \in\Cb^3 \;\; | \;\;\jb\in D, \;  \omb_{-\jb} = - \bar\omb_{\jb} , \; \jb\cdot\omb_\jb=0 \}.
\end{equation}
When $\vb$ is divergence free, this space is invariant under the Euler equations \eqref{eq:omvpde}.

To write \eqref{eq:omvpde} in Fourier space, it is convenient to define an antisymmetric matrix, the cross-product matrix, for a vector $\ab =  (a_x,a_y, a_z)^\transp$ by
\begin{equation}
	\crossmat{\ab}=\begin{pmatrix}
			0 & -a_z  & a_y \\
			a_z & 0 & -a_x \\
			-a_y & a_x & 0
			\end{pmatrix} .
\end{equation}
The definition is made so that for any vectors $\ab, \bb\in\Rb^3$, $\crossmat{\ab}\bb=\ab\times\bb$. 
Thus for example, the condition $\Omega=\nabla \times V$ implies
\begin{equation}
\label{eq:omvv}
	\omb_\jb=i\jb\times \vb_\jb = i \crossmat{\jb} \vb_\jb.
\end{equation}
The properties of the cross product imply that for any invertible $3\times3$ matrix $M$,  $M\ab\times M\bb=\text{det}(M)M^{-\transp}(\ab\times\bb)$. Thus we have 
\[
	\crossmat{M\ab}=\text{det}(M)M^{-\transp}\crossmat{\ab}M^{-1}.
\]
In the special case of a rotation matrix $R\in SO(3)$, this reduces to $\crossmat{R\ab}=R\crossmat{\ab}R^{\transp}$.

We wish to write \eqref{eq:omvpde} as an infinite set of differential equations for the dynamics of the vorticity modes $\omb_\jb$ only. Combining \eqref{eq:divfree} and \eqref{eq:omvv} implies
\[
	\jb\times\omb_\jb=-i |\jb|^2\vb_\jb.
\]
Thus from \eqref{eq:omvv}, for all $\jb\neq \mathbf{0}$,
we can take the divergence-free inverse to the curl
\begin{equation}
	\label{eq:velocitymodes}
	\vb_\jb=i\frac{\jb\times\omb_\jb}{|\jb|^2}.
\end{equation}
In addition we assume that we are in a frame of reference where $\vb_\mathbf{0} = \mathbf{0}$. 
Note also that $\omb_\mathbf{0} = \mathbf{0}$ by \eqref{eq:omvv}.
This allows us to formally write $V=(\nabla\times)^{-1}\Omega$ on the divergence-free subspace.

Since $\rho=1$, the energy per unit volume becomes
\begin{equation}
\begin{split}
	E&=\frac{\aspect_x\aspect_y\aspect_z}{2 (2\pi)^3} \int_\mathcal{D} V\cdot V\mathrm{d}\xb  
		=\frac12 \sum_{\jb\in D} \vb_\jb\cdot \vb_{-\jb} \\
		&=\frac12 \sum_{\jb \in D \setminus \zeroS} \frac{\omb_\jb \cdot \omb_{-\jb}}{|\jb|^2} , 
\end{split}
\label{eq:energy}
\end{equation}
where we used \eqref{eq:divfree} and \eqref{eq:velocitymodes} for the last step.

Now \eqref{eq:velocitymodes} implies that
\[
 \begin{split}
	(\Omega\cdot\nabla)V & =\sum_{\jb,\kb\in D\setminus \zeroS }i(\jb\cdot\omb_\kb) \vb_\jb e^{i(\jb+\kb)\cdot\xb} \\
		&=-\sum_{\jb,\kb \in D\setminus \zeroS }  (\jb\cdot\omb_{\kb})
		\frac{\jb\times\omb_\jb}{|\jb|^2}  e^{i(\jb+\kb)\cdot\xb} .
\end{split}
\] 
Similarly 
\[
 \begin{split}
	(V\cdot\nabla)\Omega& =\sum_{\jb,\kb\in D\setminus \zeroS}i(\jb\cdot\vb_{\kb})
		\omb_\jb e^{i(\jb+\kb)\cdot\xb } \\ 
		&=-\sum_{\jb,\kb\in D\setminus \zeroS} \frac{1}{|\kb|^2}\jb\cdot(\kb\times\omb_\kb)
			\omb_\jb e^{i(\jb+\kb)\cdot\xb}.
\end{split}
\]
Thus rewriting \eqref{eq:omvpde}
in Fourier space leads to the system of ODEs
\begin{equation} \begin{split}
\label{eq:3Dvorticityode}
	\dot{\omb}_\jb&
		=\sum_{\kb+\lb=\jb} \frac{1}{|\kb|^2}\left [ \lb\cdot (\kb\times \omb_{\kb}) \omb_{\lb}
		 -(\kb\cdot  \omb_\lb) \kb\times \omb_\kb \right ] \\
	 &=\sum_{\kb \in D\setminus \zeroS} A(\jb,\kb)\frac{\omb_{-\kb}}{|\kb|^2} ,
\end{split}
\end{equation}
where
\begin{equation}\label{eq:Adefine}
	A(\jb,\kb)=  \omb_{\jb+\kb}(\kb\times\jb)^{\transp}
		- (\kb\cdot \omb_{\jb+\kb} )\crossmat{\kb} .
\end{equation}
Note that since $\vb_\mathbf{0} = \omb_\mathbf{0} = \mathbf{0}$, the term $\kb = \mathbf{0}$ has been omitted in the summation so there is no singularity in \eqref{eq:3Dvorticityode}. 
Thus we derived an infinite set of ordinary differential equation for the vorticity modes that does not depend on velocity.

\section{Poisson Structure}

The Euler equations possess a Poisson structure given by 
\begin{equation}\label{eq:FunctionalPoisson}
    \{ F, G \} = \int d^3 x \, \omb \cdot \left( \left( \nabla \times \frac{\delta F}{\delta \omb} \right) \times 
    							\left( \nabla \times \frac{\delta G}{\delta \omb} \right) \right) ,
\end{equation}
on the set of divergence-free vorticity fields $\omb$, see, e.g., \cite{Zakharov97}.
Using this Poisson structure the equations of motion can be written 
$ \omb_t = \{\omb, E\}$ where $E$ is the kinetic energy of the fluid as defined in \eqref{eq:energy}.
We wish to derive a Poisson formulation for the Euler equations in Fourier space \eqref{eq:3Dvorticityode}.
As they stand they do not define a Poisson structure; for example, $A$ is not antisymmetric:
$A(\jb, \kb)^T  \not = -A(\kb, \jb)$.

To obtain a Poisson bracket for the Fourier space Euler equations \eqref{eq:3Dvorticityode} we formally rewrite  
\eqref{eq:FunctionalPoisson} in Fourier space and infer that $J(\jb, \kb)$ is defined via
\[
   \{ f, g \} = -\sum_{\jb, \kb \in D} \omega_{\jb + \kb} \cdot \left( \left( \jb \times \frac{\partial f}{\partial \omb_\jb} \right) \times \left( \kb \times \frac{\partial g}{\partial \omb_\kb} \right) \right)
= \sum_{\jb,\kb\in D} \left(\frac{\partial f}{\partial \omb_\jb}\right)^{\transp} J(\jb,\kb) \left(\frac{\partial g}{\partial \omb_\kb}\right) .
\]
Extracting $J$ from this definition gives
\begin{equation}
\label{eq:BDef}
   J(\jb, \kb) :=   \omb_{\jb + \kb} (\kb \times \jb)^T +( \jb \cdot \omb_{\jb + \kb}) \crossmat{ \kb  }  \,.
\end{equation}
We notice that $A$ and $J$ are different, 
\[
 J(\jb, \kb) - A(\jb, \kb) = ((\jb + \kb) \cdot \omb_{\jb + \kb} ) \crossmat{ \kb } \,,
\]
but their difference vanishes on the divergence-free subspace since then $(\jb+\kb) \cdot \omb_{\jb+\kb} = 0$.
Hence on the subspace $\Mcal$ defined in \eqref{eq:divfreess}, $A(\jb,\kb)=J(\jb,\kb)$ and so for initial conditions in this space, the Euler equations \eqref{eq:3Dvorticityode} can be written as
\begin{equation}
	\label{eq:3dodeB}
	\dot{\omb_\jb}=\sum_{\kb \in D  \setminus \zeroS} J(\jb,\kb)\frac{\omb_{-\kb}}{|\kb|^2}.
\end{equation}

Notice that $A$ and $J$ are matrix valued functions depending on two lattice vectors $\jb, \kb$ and a vorticity Fourier component $\omb_{\jb + \kb}$, but the dependence on the latter is suppressed in the notation.

By writing the differential equations in the form \eqref{eq:3dodeB}, the Euler equations  on a three-dimensional periodic domain 
with divergence-free vorticity are a Poisson system.  To show this we first obtain a simple lemma.

\begin{lemma}
\label{lem:JProperties}
The matrix $J$ satisfies $J(\jb,\kb) = -J(\kb,\jb)^\transp$,   $J(\jb,\kb) \kb = 0$, and $ \jb^\transp J(\jb,\kb) = 0 $.
\end{lemma}
\begin{proof}
Neither term in \eqref{eq:BDef} is antisymmetric under the exchange of  $\jb$ and $\kb$ and transposition, but their sum is.
The vector $\kb$ is clearly in the right kernel of $J$, while $\jb$ is in the left kernel since $\jb^T \crossmat{ \kb } = \jb \times \kb$.
Note in particular that these results are valid for arbitrary values of $\omb_{\jb+\kb}$. 
\end{proof}

\begin{theorem}[Simple Poisson Structure on the divergence free subspace]
\label{thm:PS3D}
The dynamics of the three-dimensional Euler equations for an incompressible, inviscid flow on a periodic domain 
are a Poisson system on the divergence-free subspace $\mathcal M$ with bracket
\begin{equation}
	\label{eq:3DBracketPROOF}
	\{f,g\}=\sum_{\jb,\kb\in D} \left(\frac{\partial f}{\partial \omb_\jb}\right)^{\transp} J(\jb,\kb) \left(\frac{\partial g}{\partial \omb_\kb}\right) ,
\end{equation}
where $J(\jb,\kb)$ is given by  \eqref{eq:BDef},
and Hamiltonian
\begin{equation}
	\label{eq:3DHamiltonianPROOF}
	H=\frac{1}{2}\sum_{\jb\in D \setminus \zeroS }\frac{\omb_{-\jb}\cdot\omb_{\jb}}{|\jb|^2} \,.
\end{equation}
The physically relevant, divergence-free subspace $\mathcal M$ given by \eqref{eq:divfreess}
 is an invariant subspace of the dynamics.
\end{theorem}

\begin{proof}
For the Hamiltonian \eqref{eq:3DHamiltonianPROOF}, the differential equations for $\omb_\jb$ are given by 
\eqref{eq:3dodeB} which, on the divergence-free subspace $\mathcal M$, is the same as 
\eqref{eq:3Dvorticityode}.

To be a Poisson bracket, \eqref{eq:3DBracketPROOF} must be  antisymmetric, bilinear, satisfy the Leibniz rule, and satisfy the Jacobi identity \cite{Morrison82,Olver00,Meiss17}. The first property follows easily since
\[ \begin{split}
	\{g,f\}&=\sum_{\kb,\jb\in D} \left(\frac{\partial g}{\partial \omb_\kb}\right)^{\transp} J(\kb,\jb) \left(\frac{\partial f}{\partial \omb_\jb}\right) \\
		&=-\sum_{\kb,\jb\in D} \left(\frac{\partial g}{\partial \omb_\kb}\right)^{\transp} J(\jb,\kb)^{\transp} \left(\frac{\partial f}{\partial \omb_\jb} \right) \\
		&=-\sum_{\kb,\jb\in D} \left (\left(\frac{\partial f}{\partial \omb_\jb}\right)^{\transp} J(\jb,\kb) \left(\frac{\partial g}{\partial \omb_\kb}\right ) \right )^{\transp} \\
		&=-\{f,g\} ,
\end{split}\]
as $J(\kb,\jb)^{\transp} = -J(\jb,\kb)$  by Lemma \ref{lem:JProperties}. The bracket is also bilinear as the bracket is linear in the partial derivatives of $f$ and $g$.  The Leibniz rule holds due the bracket being linear in the derivative of $f$.

The final condition is the Jacobi identity, $\{f,\{g,h\}\} + \{g,\{h,f\}\} + \{h,\{f,g\}\} = 0$. For a bracket of the form \eqref{eq:3DBracketPROOF}, this can be reduced to a condition on the structure matrix, see \cite[Theorem 2]{Morrison82} and \cite[Eq. (6.15)] {Olver00}, which becomes $Z_{\alpha, \beta, \gamma}(\ib,\jb,\kb) = 0$ where
\begin{equation}\label{eq:Jacobi}
\begin{split}
	Z_{\alpha, \beta, \gamma}(\ib,\jb,\kb) =&\sum_{\delta\in\{x,y,z\},\lb\in D}
	\bigg [ J(\ib,\lb)_{\alpha,\delta} \frac{\partial J(\jb,\kb)_{\beta,\gamma}}{\partial (\omega_{\lb})_\delta}\\
	 &\;\;+J(\kb,\lb)_{\gamma,\delta} \frac{\partial J(\ib,\jb)_{\alpha,\beta}}{\partial (\omega_{\lb})_\delta}
	+J(\jb,\lb)_{\beta,\delta} \frac{\partial J(\kb,\ib)_{\gamma,\alpha}}{\partial (\omega_{\lb})_\delta} \bigg ]\\
	=&\sum_{\delta\in\{x,y,z\}}
	\bigg [ J(\ib,\jb+\kb)_{\alpha,\delta} \frac{\partial J(\jb,\kb)_{\beta,\gamma}}{\partial (\omega_{\jb+\kb})_\delta}\\
	 &\;\;+J(\kb,\ib+\jb)_{\gamma,\delta} \frac{\partial J(\ib,\jb)_{\alpha,\beta}}{\partial (\omega_{\ib+\jb})_\delta}
	+J(\jb,\kb+\ib)_{\beta,\delta} \frac{\partial J(\kb,\ib)_{\gamma,\alpha}}{\partial (\omega_{\kb+\ib})_\delta} \bigg ] ,
\end{split}
\end{equation}
for all $\alpha,\beta,\gamma\in \{x,y,z\}$ and $\ib,\jb,\kb\in D$. Here we have used $x,y,z$-subscripts to denote the components, e.g., 
\[
\omb_\jb=((\omega_\jb)_x,(\omega_\jb)_y,(\omega_\jb)_z),
\]
and similarly for the matrix $J$. 
By a tedious computation one can show that each non-zero entry of $Z_{\alpha, \beta, \gamma}$ is 
proportional to $(\ib+\jb+\kb) \cdot \omb_{\ib+\jb+\kb}$ and hence vanishes on $\mathcal M$.

Finally, the subspace $\mathcal M$ is an invariant subspace because Lemma~\ref{lem:JProperties} gives
\begin{equation}\label{eq:InvSubSpace}
\frac{\mathrm{d}}{\mathrm{d} t}  (\jb \cdot \omb_\jb) = \{ \jb \cdot \omb_\jb , H \} = \sum_{\jb, \kb} \jb^\transp J(\jb, \kb) \frac{\partial H}{\partial \omb_k} = 0 ,
\end{equation}
assuming $ \jb \cdot \omb_\jb  = 0$ for all $\jb$ .

Thus the system given by \eqref{eq:3DBracketPROOF}, \eqref{eq:3DHamiltonianPROOF} and \eqref{eq:divfreess} is a Poisson system, and generates the dynamics of the Euler equations on a three-dimensional periodic domain.
\end{proof}

The Fourier space Poisson structure \eqref{eq:BDef} is  ``tainted'' in the terminology of \cite{Chandre13}: 
it does not satisfy the Jacobi identity in full space obtained by dropping the 
divergence-free condition from $\mathcal M$. 
However, following the ideas of \cite{Chandre13} it is simple to modify the Poisson structure $J$ 
so that is satisfies the Jacobi identity on the full space using the projector $\omb \mapsto \omb - \nabla \Delta^{-1} \nabla \omb$
which in Fourier space becomes 
\[
{\mathcal P} \omb_{\jb} := -\jb \times ( \jb \times \omb_\jb) \frac{1}{|\jb|^2} =  (id - \jb \jb^T / |\jb|^2) \omb_{\jb}.
\]
In the following we write $J(\jb, \kb) = J(\jb, \kb, \omb_{\jb+\kb})$ to make the dependence on $\omb_{\jb + \kb}$ 
explicit. 
\begin{theorem}[Projected Poisson Structure]
\label{thm:Jprojected}
The matrix valued function 
\[
   {\mathcal J}(\jb,\kb,\omb_{\jb+\kb}) = J( \jb, \kb, {\mathcal P} \omb_{\jb+\kb})
\]
defines a Poisson structure with bracket as in \eqref{eq:3DBracketPROOF}
not only on the divergence free subspace but for all $\omega_\jb \in \Cb^3$, $\jb \in D$. 
\end{theorem}
\begin{proof}
Antisymmetry and bilinearity are inherited from $J$. It remains to check the 
Jacobi identity. Previously we noted that all non-zero terms in \eqref{eq:Jacobi} are proportional to $\lb \cdot \omb_\lb$ where $\lb = {\ib + \jb + \kb}$.
By linearity of $\mathcal J$ in $\omb$, this becomes $\lb \cdot {\mathcal P} \omb_\lb = -\lb \cdot (\lb \times ( \lb \times \omb_\lb)) = 0$ and the 
Jacobi identity holds.
\end{proof}

It may seem as if the functions $\jb \cdot \omb_\jb$ are Casimirs of $J$; indeed if $J$ were to 
satisfy the Jacobi identity in the full space then  this would be the case by Lemma~\ref{lem:JProperties}. 
However, outside the subspace where $\jb \cdot \omb_\jb  =0 $ for all $\jb \in D$ the Jacobi identity 
does not hold for $J$, so it is not a Poisson structure outside of the manifold $\mathcal{M}$. By contrast, for $\mathcal{J}$ the 
Jacobi identity is satisfied everywhere and  the functions  $\jb \cdot \omb_\jb$ are true Casimirs. 
Recall that a function $f$ is a Casimir of a bracket if $\{f,g\} = 0$ for any function $g$. 
If $f$ is a Casimir then it is an invariant for  \textit{any} dynamics generated by the bracket, 
independent of the choice of Hamiltonian. For the modified bracket $\mathcal J$ any subspace defined 
by setting the Casimirs $\lb \cdot \omb_\lb  $ a constant is an invariant subspace. 
In particular the divergence-free subspace $\mathcal M$ defined in \eqref{eq:divfreess} is 
such an invariant subspace defined by setting all Casimirs to zero. 

In the next section we will explicitly reduce the Poisson structure to the divergence-free subspace by employing a local transformation that allows for the explicit projection and restriction to the invariant subspace. This reduction works both for $J$ and $\mathcal J$, 
since they are equal on the invariant subspace \eqref{eq:divfreess}.

\section{Poisson Structure on the Divergence-Free Subspace}\label{sec:divfreess}

We wish to explicitly restrict our system to the manifold $\mathcal{M}$ \eqref{eq:divfreess}, the divergence-free subspace where $\jb\cdot\omb_\jb=0$ for all $\jb$. We will write the Poisson system as a two-dimensional, unrestricted system, thereby reducing the size of the system by one-third.

To do so, we introduce a rotation matrix $R_\jb \in SO(3)$ for each $\jb\in\Rb^3$ 
with the goal of making $\jb \cdot \omb_\jb$ the first component of a rotated vorticity; i.e., we 
rotate each $\jb \in \Rb^3$ to a vector parallel to the $x$-axis. 
In principle we could assign a different rotation matrix to each lattice point $\jb$.
Such a rotation matrix can be constructed using a non-zero vector $\nb\in\mathbb{R}^3$ and set
\begin{equation}\label{eq:Rotation}
	R_\jb = \begin{pmatrix} \jb^\transp / |\jb| \\
						    (\jb \times \nb)^\transp / |\jb \times \nb| \\
						    (\jb \times ( \jb \times \nb))^\transp /(| \jb \times \nb | | \jb |)
			\end{pmatrix} 
		= 	
 \begin{pmatrix} \jb^\transp / |\jb| \\
						    (\jb \times \nb)^\transp / |\jb|_2 \\
						    (\jb \times ( \jb \times \nb))^\transp /(| \jb  |_2 | \jb |)
			\end{pmatrix} 		
			\,,   \quad \jb \times \nb \neq 0,
\end{equation}
so that the rows of $R_\jb$ are orthonormal. Here we defined a projected ``2-norm" $|\jb|_2 = | \jb \times \nb|$.
Notice that we have 
\begin{equation}\label{eq:Sdefine}
   R_{-\jb} R_\jb^\transp = S = \mathrm{diag}( -1, -1, 1)\,.
\end{equation}
The ``signature" matrix $S$ will remain in the transformed Hamiltonian. 

In the most general case the vector $\nb$ in $R_j$ can be any vector that is not parallel to $\jb$.
Perhaps the simplest and most convenient assumption would be to choose $\nb = \hat{e}_x$, 
the unit vector in the $x$-direction.
However, when $\jb$ and $\nb$ are parallel, namely for $\jb = a \hat{e}_x$, this definition fails. 
To allow for this, we define
\[
	R_{a\nb}=\mathbb{I} , \quad R_{-a\nb} = S \mbox{ for all }a > 0.
\]
We then define a rotated vorticity by
\begin{equation}\label{eq:omcheck}
     \check{\omb}_\jb=R_{\jb }\omb_\jb ,
\end{equation}
so that $ (\check{\omb}_\jb)_x \propto \jb^{\transp} \omb_\jb = 0$  on the divergence-free subspace. Note that this definition implies the reality condition $\check{\omb}_{-\jb} =  S \overline{ \check{\omb}}_\jb$.

Instead of this discrete choice for $\nb$, one may try to define the special cases for $R_\jb$ by a limit where $\nb$ tends to a vector parallel to $\jb$. However, the result  would then depend on how the limit is taken, and this is why we resorted to the definition of $R_{\pm a\nb}$ above.

It might have seemed convenient to choose a matrix so that
$R_{-\jb}$ and $R_\jb^\transp$ are inverses of each other so that $S$ would become the identity matrix, unlike the relation \eqref{eq:Sdefine}. In this case $H$, \eqref{eq:3DHamiltonianPROOF}, would be invariant under the transformation \eqref{eq:omcheck}. 
Indeed it is possible to change the definition of $R_\jb$ to make $S = \mathbb{I}$. 
For this one has to divide the set $D$ into ``positive" and ``negative" lattice 
points, and use this sign in the definition of $R_\jb$; see \cite{Worthington17} for the details.

Since the lattice $D$, \eqref{eq:lattice}, is discrete, one could in principle
use a different $R_\jb$ for each lattice point, as long as the first row is proportional to $\jb$.
Specifically, each $R_\jb$ could include an additional rotation about the $\jb$-axis
by an angle that depends on $\jb$. For simplicity,  we assume here that all these angles
are equal to 0.

Since the first row of $R_\jb$ is in the $\jb$ direction, the first element of $\check{\omb}_\jb$ must be zero on the divergence-free subspace, we can project down to this subspace and still capture the full dynamics. 

\begin{proposition}[Restricted Poisson structure]
\label{prop:redstr}
Define $\tilde{\omb}_\jb = (\check{\omega}_{\jb,y}, \check{\omega}_{\jb,z})$ and the reduced manifold
\begin{equation}
	\tilde{\mathcal{M}}=\{ \tilde{\omb}_\jb\in\Cb^2 \;|\; \tilde{\omb}_{-\jb} = \tilde S \overline{ \tilde{\omb}}_\jb \; ,\; \jb\in D\},
\end{equation}
with $\tilde S = \mathrm{diag}( -1, 1)$.
The Poisson Structure of Theorem \ref{thm:PS3D} reduces to
\begin{equation}
	\label{eq:3DBracketReduced}
	\{f,g\}=\sum_{\jb,\kb\in D} \left(\frac{\partial f}{\partial \tilde{\omb}_\jb}\right)^{\transp} \tilde{J}(\jb,\kb) \left(\frac{\partial g}{\partial \tilde{\omb}_\kb}\right)
\end{equation}
for a matrix valued function $\tilde{J}: \tilde{\mathcal{M}}^* \to \tilde{\mathcal{M}}$ which depends linearly on $\tilde{\omb}_{\jb}$. 
The new Hamiltonian becomes
\begin{equation}
	\label{eq:3DHamiltonianReduced}
	\tilde{H}=
	\frac{1}{2}\sum_{\jb\in D \setminus \zeroS }\frac{1}{|\jb|^2}
			{\tilde{\omb}_{-\jb}^\transp \tilde S \tilde{\omb}_{\jb}}	.		
\end{equation}
Then the dynamics of 
\begin{equation}\label{eq:projode}
	\dot{\tilde{\omb}}_\jb=\sum_{\kb\neq \mathbf{0}}\tilde{J}(\jb,\kb)  \tilde S \frac{\tilde{\omb}_{-\kb}}{|\kb|^2}
\end{equation}
are equivalent to the dynamics of \eqref{eq:3dodeB} on the divergence-free subspace.
\end{proposition}

\begin{proof}
	Using the matrix \eqref{eq:Rotation} $R_\jb$ and the definition \eqref{eq:omcheck} for $\check{\omb}_\jb=(\check{\omega}_{\jb,x},\check{\omega}_{\jb,z},\check{\omega}_{\jb,y})$, the divergence-free condition implies that $\check{\omega}_{\jb,x}$ is constant and zero for all $\jb$. We first look, however, at the equations for all three components of $\check{\omb}_\jb$.

Define  
\begin{equation}
	\check{J}(\jb,\kb)  =R_\jb {J}(\jb,\kb) R_\kb^\transp 
	     =R_\jb\bigg [ (R_{\jb+\kb}^{\transp}\check{\omb}_{\jb+\kb})(\kb \times \jb)^{\transp} 
	     +( \jb\cdot( R_{\jb+\kb}^{\transp}\check{\omb}_{\jb+\kb} ))\crossmat{\kb} \bigg ] R_{\kb}^{\transp}.
\end{equation}

Then the dynamics of $\check{\omb}_\jb$ are
\begin{equation}
	\dot{\check{\omb}}_\jb=\sum_{\kb \neq 0} \check{J}(\jb,\kb)  S \frac{\check{\omb}_{-\kb}}{|\kb|^2}.
\end{equation}
Since the first column of $R_\kb^\transp$ is $\kb$ and the first row of $R_\jb$ is $\jb^\transp$, we can conclude by Lemma \ref{lem:JProperties} that when $\omb_{\jb+\kb} \in \Mcal$,
\[
	\check{J}(\jb,\kb)_{x,\beta}=\check{J}(\jb,\kb)_{\beta,x}=0, \quad \beta \in \{x,y,z\},	
\]
which is expected since the divergence-free subspace is invariant. 
This remains true in the most general case where a different vector $\nb$ is used for the matrices 
$R_\jb$, $R_\kb$, and $R_{\jb + \kb}$. We already noted after Thm.~\ref{thm:PS3D} that the conditions $\jb \cdot \omb_\jb = 0$ are invariant under the dynamics, and now we have explicitly reduced the dynamics to the invariant submanifold.

Thus we can ignore the dynamics of the coordinate $\check{\omega}_{\jb,x}$, and define
$\tilde{\omb}_\jb = (\check{\omega}_{\jb,y}, \check{\omega}_{\jb,z})$  and  the $2 \times 2$
Poisson matrix
\begin{equation}
	\tilde{J}(\jb,\kb)=\begin{pmatrix} \check{J}(\jb,\kb)_{y,y} & \check{J}(\jb,\kb)_{y,z} \\
			\check{J}(\jb,\kb)_{z,y} & \check{J}(\jb,\kb)_{z,z}
			\end{pmatrix}.
	\end{equation}
Note that the reality condition for $\tilde{\omb}_\jb$ becomes $\tilde{\omb}_{-\jb} = \tilde S \overline{ \tilde{\omb}}_\jb$.	
The Hamiltonian in the new coordinates is obtained by substituting $\tilde{\omb}_{\jb}$ into \eqref{eq:3DHamiltonianPROOF} to obtain \eqref{eq:3DHamiltonianReduced}.
	
The form of $\tilde{J}$ is quite complicated. 
We can directly check that it satisfies the conditions for a Poisson structure matrix (antisymmetry and Jacobi identity). 
However, this also follows from the fact that the matrix $\check{J}$ has vanishing first row and column, so that the Jacobi-identity 
is inherited from the matrix $J$, in particular because the transformation to the new coordinates is linear. Similarly,
the antisymmetry is not changed by a linear transformation.

As the coordinates $\tilde{\omb}_\jb$ are no longer restricted, the Poisson manifold is the full space
$\tilde{\mathcal{M}}$. Then the Poisson bracket with structure matrix $\tilde{J}(\jb,\kb)$ and the Hamiltonian $\tilde{H}$ generate the dynamics of the Euler equations in the new coordinates $\tilde{\omb}_\jb$.
\end{proof}

Before we list the general formulas for $\tilde J$ we give the three special cases that occur when either $\jb$, $\kb$, or $\jb + \kb$ are proportional to 
$\nb = \hat{e}_x$.  For example, when $\jb = s a \nb$ the corresponding rotation matrix is replaced by $\mathbb{I}$ or $S$, depending on whether $s = 1$ or $s = -1$, respectively.
Analogous rules hold for $\kb$ and $\jb + \kb$. The results are
\[
   \jb || \nb: \quad \tilde J =  \frac{j_x}{|\jb+\kb|}  \begin{pmatrix}
       (j_x k_y \tilde{\omb}_{\jb+\kb,z}+k_z \tilde{\omb}_{\jb+\kb,y} |\jb+\kb|) s & k_z \tilde{\omb}_{\jb+\kb,z} |\kb| s \\
       j_x k_z \tilde{\omb}_{\jb+\kb,z}-k_y \tilde{\omb}_{\jb+\kb,y} |\jb+\kb| & -k_y \tilde{\omb}_{\jb+\kb,z} |\kb| 
   \end{pmatrix},
\]
\[
   \kb || \nb: \quad \tilde J =  \begin{pmatrix}
      -(j_x+k_x) (j_y \tilde{\omb}_{\jb+\kb,z}-j_z \tilde{\omb}_{\jb+\kb,y} s) & |\kb| (j_z \tilde{\omb}_{\jb+\kb,z}+j_y \tilde{\omb}_{\jb+\kb,y} s) \\
      |\jb| (j_z \tilde{\omb}_{\jb+\kb,z}+j_y \tilde{\omb}_{\jb+\kb,y} s) & 0 
   \end{pmatrix},
\]
\[
   \jb+\kb || \nb: \quad \tilde J = \frac{k_x}{|\jb+\kb|}  \begin{pmatrix}
      -(j_y k_x \tilde{\omb}_{\jb+\kb,z}+j_z \tilde{\omb}_{\jb+\kb,y} |\jb+\kb|) s & j_y \tilde{\omb}_{\jb+\kb,y} |\jb+\kb|-j_z k_x \tilde{\omb}_{\jb+\kb,z} \\
      -j_z \tilde{\omb}_{\jb+\kb,z} |\jb| s & j_y \tilde{\omb}_{\jb+\kb,z} |\jb| 
   \end{pmatrix} \,.
\]

Now we proceed to the general case where none of $\jb$, $\kb$, $\jb+\kb$ are parallel to $\nb$.
Since $\tilde{J}(\jb,\kb)$ is linear in $\tilde{\omb}_{\jb+\kb}$, and the $x$ component of $\tilde {\omb}_{\jb+\kb}$ is zero, we can write
\begin{equation}
\begin{split}
	\tilde{J}(\jb,\kb)=&\begin{pmatrix} \tilde{J}_y(\jb,\kb)_{y,y} & \tilde{J}_y(\jb,\kb)_{y,z} \\
			\tilde{J}_y(\jb,\kb)_{z,y} & \tilde{J}_y(\jb,\kb)_{z,z}
			\end{pmatrix}\tilde{\omega}_{\jb+\kb,y}\\
				&+\begin{pmatrix} \tilde{J}_z(\jb,\kb)_{y,y} & \tilde{J}_z(\jb,\kb)_{y,z} \\
			\tilde{J}_z(\jb,\kb)_{z,y} & \tilde{J}_z(\jb,\kb)_{z,z}
			\end{pmatrix}\tilde{\omega}_{\jb+\kb,z}
\end{split}
\end{equation}
where $\tilde{J}_y(\jb,\kb)_{i,j}$, $\tilde{J}_z(\jb,\kb)_{i,j}$ are functions of $\jb, \kb$ only. 

\begin{lemma}
The Poisson structure matrix $\tilde J$ after restriction to the divergence-free subspace satisfies
\begin{equation}
\label{eq:reqcond}
\begin{split}
\frac{1}{|\kb|}\tilde{J}_y(\jb,\kb)_{b,y}
+\frac{1}{|\jb+\kb|}\tilde{J}_z(\jb,-\jb-\kb)_{b,z}=0,\\
\frac{1}{|\kb|}\tilde{J}_y(\jb,\kb)_{b,z}
+\frac{1}{|\jb+\kb|}\tilde{J}_y(\jb,-\jb-\kb)_{b,z}=0,\\
\frac{1}{|\kb|}\tilde{J}_z(\jb,\kb)_{b,y}
+\frac{1}{|\jb+\kb|}\tilde{J}_z(\jb,-\jb-\kb)_{b,y}=0,
\end{split}
\end{equation}
for $b=y$ and $b=z$.
\end{lemma}
\begin{proof}
This follows by direction computation from the explicit formulas given below.
\begin{equation}
	\tilde{J}_y(\jb,\kb)_{y,y}=
\nb \cdot  \Big(  ( \jb \times \kb) \times  ( \kb  \, |\jb |_2^2  +  \jb  \, |\kb|_2^2 )\Big)
\frac{1 }{|\jb|_2 |\kb|_2 |\jb + \kb|_2} \,,
\end{equation}
\begin{equation}
	\tilde{J}_y(\jb,\kb)_{y,z}=  -\nb \cdot  ( \jb\times \kb )  \frac{ |\jb|_2 |\kb| }{ |\jb + \kb|_2 |\kb|_2}
\end{equation}
\begin{equation}
	\tilde{J}_y(\jb,\kb)_{z,y}=  -\nb \cdot  ( \jb\times \kb )   \frac{ |\kb|_2 |\jb| }{ |\jb + \kb|_2 |\jb|_2} 
\end{equation}
\begin{equation}
	\tilde{J}_y(\jb,\kb)_{z,z}=0,
\end{equation}
\begin{equation}
       \tilde{J}_z(\jb,\kb)_{y,y} = \nb \cdot (\jb \times \kb)\,  |\nb \times ( \jb \times \kb)|^2 \frac{1}{|\jb|_2 |\kb|_2 |\jb+\kb|_2 |\jb + \kb| },
\end{equation}
\begin{equation}
 \tilde{J}_z(\jb,\kb)_{y,z}= - \frac{|\kb|}{|\jb+\kb|} \tilde J_y(\jb, -\jb-\kb)_{y,y},
\end{equation}
\begin{equation}
	\tilde{J}_z(\jb,\kb)_{z,y}=-\tilde{J}_z(\kb,\jb)_{y,z},
\end{equation}

\begin{equation}
\tilde{J}_z(\jb,\kb)_{z,z}=  -\nb \cdot (  \jb\times \kb)  \frac{|\jb||\kb| |\jb+\kb|_2}{|\jb|_2 |\kb|_2 |\jb+\kb|}.
\end{equation}
\end{proof}

The condition for antisymmetry now reads
\begin{equation}
\tilde{J}_y(\jb,\kb) +\tilde{J}_y(\kb,\jb)^T= \tilde{J}_z(\jb,\kb) +\tilde{J}_z(\kb,\jb)^T = 0 \,.
\end{equation}

\section{Helicity}

The three-dimensional Euler equations have a constant of motion, called the \emph{helicity}, 
defined by 
\[
	h =\frac{\aspect_x\aspect_y\aspect_z}{(2\pi)^3} \int_\mathcal{D} V\cdot (\nabla \times V)\;\mathrm{d}x \\
	  =\frac{\aspect_x\aspect_y\aspect_z}{(2\pi)^3} \int_\mathcal{D} V\cdot \Omega\;\mathrm{d}x,
\]
see, e.g., \cite{Gibbon08}. 
In terms of the Fourier modes on the periodic domain, this quantity becomes
\[
	h = \sum  \vb_\kb \cdot \omb_{-\kb}  \\
      = \sum \frac{ i  }{ |\kb|^2} \kb \cdot (\omb_{\kb} \times \omb_{-\kb}) .
\]
Note that $h\in\mathbb{R}$, as $\omb_{-\kb}=\overline{\omb}_{\kb}$ so $\omb_\kb\times \omb_{-\kb}\in i\mathbb{R}^3$.

\begin{lemma}[The helicity is a Casimir]
\label{lem:helicity}
The helicity $ \hel$ given by \eqref{eqn:helicity} is a Casimir of the projected Poisson structure  $\mathcal{J}$ 
described in Theorem~\ref{thm:Jprojected}.
\end{lemma}
\begin{proof}
To show that $\hel$ is a Casimir we need to show that $\{ g, \hel \}  = 0$ for any function $g$.
This requires that $\nabla_\omega h$ is in the kernel of $\mathcal{J}$. 
The gradient of $h$ is given by 
\[
 \frac{\partial h}{\partial \omb_\kb}  = -\frac{2i}{|\kb|^2} \kb \times \omb_{-\kb}  \,.
\]
and so we must have
\begin{equation}
\label{eqn:hsum}
   -i \sum_{\kb \in D} \mathcal{J}(\jb, \kb)  \frac{1}{|\kb|^2} \kb \times \omb_{-\kb} = 0.
\end{equation}
This is true because of the identity
\begin{equation}
\label{eqn:calJid}
    \mathcal{J}(\jb,\kb, \omb_{\jb+\kb}) \frac{1}{|\kb|^2}\kb \times \omb_{-\kb} 
 + \mathcal{J}(\jb,-\jb-\kb, \omb_{-\kb}) \frac{1}{|\jb+\kb|^2} (-\jb-\kb) \times \omb_{\jb+\kb} = 0
\end{equation}
so that in the sum \eqref{eqn:hsum} pairs of terms with index $\kb$ and index $-\jb - \kb$ cancel.
\end{proof}

Now we are going to show that a similar result holds for the Poisson structure restricted to the divergence free subspace.
By transforming to the coordinates introduced in Section \ref{sec:divfreess} and noting that
\begin{equation}
\begin{split}
\omb_{\kb} \times \omb_{-\kb} 
&= R^\transp_{\kb} \check \omb_{\kb}  \times R^\transp_{-\kb} \check \omb_{-\kb} \\
            & = R^\transp_{\kb} (\check \omb_{\kb} \times S \check \omb_{-\kb})  \\
            & = R^\transp_{\kb} ( \check \omb_{\kb} \times \overline{\check{\omb}}_{\kb} ) \\
	& =  i  \frac{ \kb} {|\kb|}  2 \Im ( \tilde\omega_{\kb, y} \overline{\tilde\omega}_{\kb,z} )
          = \frac{ \kb} {|\kb|} ( \tilde\omega_{\kb,y} \tilde \omega_{-\kb,z} + \tilde \omega_{\kb,z} \tilde \omega_{-\kb,y}),
\end{split}
\end{equation}
we see that the helicity can be rewritten as
\begin{equation}\label{eqn:helicity}
   \tilde \hel = \sum_\kb \frac{1}{|\kb|}  2 \Im ( \overline{\tilde\omega}_{\kb, y} {\tilde\omega}_{\kb,z} ) .   
\end{equation}
We now confirm that this is a Casimir of the Poisson system described in Proposition \ref{prop:redstr}. 
\begin{lemma}
\label{lem:helicity}
The helicity $\tilde  \hel$ given by \eqref{eqn:helicity} is a Casimir of the restricted Poisson structure described in Proposition~\ref{prop:redstr}.
\end{lemma}
\begin{proof}
For $\tilde \hel$ to be a Casimir of $\tilde J$ 
the sum  $\sum_{\kb \in D} \tilde{J}(\jb,\kb)\frac{\partial \tilde h}{\partial \tilde{\omega}_\kb}=0$ for all $\jb\in D$.
 Now, $\frac{\partial \tilde h}{\partial \tilde{\omega}_\kb}=\frac{2i}{|\kb|}\begin{pmatrix} \tilde \omega_{-\kb,z} \\ \tilde \omega_{-\kb,y}\end{pmatrix}$. 
By left-multiplying by $\tilde{J}(\jb,\kb)$, and separating out the terms, the required conditions are exactly those stated in \eqref{eq:reqcond}. We can thus conclude that the helicity is a Casimir.
\end{proof}

It is clear that the identities \eqref{eq:reqcond} are the projected component wise version of 
the identity \eqref{eqn:calJid}.

\section{An Example: Shear Flows}

A shear flow on the torus is an equilibrium solution $V^{e}$ of Euler's equations with the direction of the velocity constant and  $(V^{e} \cdot \nabla )V^{e} = 0$.
The vorticity $\Omega^{e} = \nabla \times V^{e}$ then must have the form 
\[
   \Omega^e(\xb) = \Gb C( \pb \cdot \xb) 
\]
where $C$ is an arbitrary periodic function with period $2\pi$, $\pb = (p_1, p_2, p_3)^\transp$ is a fixed vector with components that are co-prime integers, and $\Gb$ is a real fixed vector satisfying $\Gb \cdot \pb = 0$.
Let the Fourier coefficients of $C$ be $c_n$, 
then the corresponding Fourier modes are $\omb_{\kb}^e = \Gb c_n \delta_{\kb - n \pb}$.
Thus 
\[
   \left.\dot \omb_\jb\right|_{\omb^e} = 
   \sum_\kb J(\jb, \kb, \Gb c_n \delta_{\kb + \jb - n \pb})  \frac{\Gb c_{-m}  \delta_{ -\kb + m \pb} }{|\kb|^2}
\]
where we have written $J$ with three arguments to make the linear dependence on the last argument 
explicit.
The Kronnecker $\delta$'s pick out the terms $\kb = m \pb$ and $\jb = (n - m ) \pb$.
Relabelling $n \to n+m $, the non-zero equations become
\[
   \left. \dot \omb_{n \pb}\right|_{\omb^e}  = \sum_m J(n \pb, m \pb, \Gb c_{m+n}) \frac{\Gb  c_{-m } }{ |m\pb|^2} \,.
\]
Now $J$ vanishes because $\pb \cdot \Gb = 0$ (the divergence-free condition) and because $\jb=n\pb$ and $\kb=m\pb$ are parallel.
Hence $\Omega^e$  is indeed an equilibrium, since either $J$ vanishes or is multiplied by zero coming from the 
gradient of the Hamiltonian.
A similar argument applies to the projected Poisson structure $\mathcal{J}$ and the restricted Poisson structure $\tilde J$.

A point of a Poisson structure at which the co-rank of the Poisson 
tensor is less than the maximal co-rank is called a \emph{singular point} of the Poisson structure. 
If the equilibrium is a regular point of the Poisson structure
then the gradient of the Hamiltonian is a linear combination of 
gradients of the Casimirs, proving that the gradient of the Hamiltonian is in 
the kernel of the Poisson structure, and hence the vector field vanishes at the equilibrium point.
For example, this happens for 
the two-dimensional case and has been used in \cite{Dullin2016} to study the nonlinear 
stability of two-dimensional shear flows through the Energy-Casimir method.

In the three-dimensional case at hand here, the gradient of the Hamiltonian $ H$
is not proportional to the gradient of the Casimir $ \hel$, as we now show.
At the equilibrium point the gradient $\nabla_\omega  H$ has components
\[
   \left. \frac{ \partial  H}{\partial  \omb_\kb} \right|_{ \omb^e} = 
   		\frac{1}{|n \pb|^2}  \Gb c_{-n} \delta_{ \kb + n \pb}
\]
The gradient of $ \hel$ at the equilibrium is 
\[
   \left. \frac{ \partial  h}{\partial  \omb_\kb} \right|_{ \omb^e} = 
   		\frac{2i}{|n \pb|^2}  n \pb \times \Gb c_{-n} \delta_{ \kb + n \pb}
\]
Clearly the two gradients are not proportional.
Assuming that there is no additional (unknown) Casimir function
this implies that the equilibrium point is a singular point of the Poisson structure, 
i.e.\ a point where the co-rank drops.
A similar calculation with the same result can be done for the restricted Poisson structure $\tilde J$.


When the gradient of the Hamiltonian is a linear combination of gradients 
of Casimirs the linearisation of the vector field can be found 
by multiplying the Poisson structure at the equilibrium by a certain Hessian.
This is, however, not true at a singular point of the Poisson structure.
We refrain from giving explicit formulas for the linearisation here, but we conclude this example
by observing that because periodic shear flows in three dimensions are singular points 
of the Poisson structure the Energy-Casimir method cannot be applied in this case.
This is a notable difference from shear flows in two dimensions.


\section{Conclusion}

By describing the dynamics of the vorticity Fourier modes as a Poisson system, we have opened the possibility of future study analysing and exploiting this structure.

One possible use for this structure could be the development of a structure preserving integrator akin to the integrator for the two-dimensional Euler equations in \cite{Mclachlan93}. Even though that integrator in practice requires the finite-dimensional truncation developed in \cite{Zeitlin90}, one may hope to also find such a truncation and hence a similar integrator for the three-dimensional problem.
In the two-dimensional case the analogue of our $\tilde J$ (or $\mathcal{J}$) is simply the $z$-component of the 
cross product of $ \jb \times \kb$, and the finite-dimensional truncation leads to the sine-bracket \cite{Zeitlin90}.
Our Poisson structure $\tilde J$ for the three-dimensional case is already so much more complicated
algebraically that we were not able to find an analogue of the sine-bracket.
Nevertheless, we suspect that such a structure preserving truncation exists.

The projected Poisson structure $\mathcal{J}$ and the restricted  Poisson structure $\tilde J$ 
can also be the basis for nonlinear stability analysis.
In the two-dimensional case  \cite{Dullin2016} nonlinear stability results were obtained 
using the Energy-Casimir method.
But we have shown above that for periodic shear flows in three dimensions the Energy-Casimir method
fails, because the Poisson structure is singular at the equilibrium. 
It would be interesting to find other equilibria for which the gradient of the Hamiltonian $ H$ 
and the gradient of the helicity $ \hel$ are proportional (and hence these are regular points of 
the Poisson structure), and for such equilibria the Energy-Casimir method may work.
It has been suggested that Beltrami flows may provide such examples.

 \section*{Acknowledgements} 
JDM acknowledges support from the US National Science Foundation under grant DMS-1812481.
This research was completed while HRD and JDM were in residence at the Mathematical Sciences Research
Institute in Berkeley, California, during the Fall 2018 semester, supported U.S. National Science Foundation
under grant DMS-1440140.
We would like to thank the two anonymous referees whose comments and questions have greatly improved the manuscript.

\bibliographystyle{amsalpha}      
\bibliography{thesis_bib}{}       

\end{document}